\documentclass{article}

\title{Sink equilibria and the attractors of learning in games}
\usepackage{times}

\usepackage{graphicx}
\usepackage{amssymb,amsmath}
\usepackage{amsthm}
\usepackage{mathtools}
\usepackage{booktabs}
\usepackage{paralist}
\usepackage{tikz-cd}
\usepackage{tikz-network}
\usepackage{caption}
\usepackage{mathdots}
\usepackage{csquotes}
\usepackage{todonotes}
\usepackage{xargs}
\usepackage{comment}
\usepackage{standalone}
\usepackage{sgame}
\usepackage{bm}
\usepackage{natbib}
\usepackage{url}

\usepackage[total={6in, 8in}]{geometry}

\usepackage{subcaption}

\tikzset{ball/.style={circle, draw, fill=black,inner sep=0pt, minimum width=4pt}}

\usepackage{pgfplots}
\pgfplotsset{compat = newest}

\usepackage{xcolor}
\usetikzlibrary{decorations.markings}

\usepackage{tikz}
\usetikzlibrary{arrows.meta}
\tikzset{label/.style = {inner sep=1pt, fill=white}}
\tikzset{nd/.style={inner sep=1pt}}
\tikzset{>=Latex}
\tikzset{arc/.style = {->, semithick, >=Latex}}

\newtheorem{thm}{Theorem}[section]

\newtheorem{defn}[thm]{Definition}

\newtheorem{lem}[thm]{Lemma}

\newtheorem{conj}[thm]{Conjecture}
\newtheorem{prop}[thm]{Proposition}

\newcommand{\real}{\mathbb{R}}

\newcommand{\exptu}{\mathbb{U}}
\newcommand{\ddt}{\frac{\mathrm{d}}{\mathrm{d}t}}

\newcommand{\arc}[3][]{ #2 \xlongrightarrow{#1} #3}

\DeclareMathOperator{\intr}{int}
\DeclareMathOperator{\bd}{bd}
\DeclareMathOperator{\content}{content}

\author{Oliver Biggar and Christos Papadimitriou}
\date{}

\begin{document}

\maketitle

\begin{abstract}
\noindent Characterizing the limit behavior---that is, the attractors---of learning dynamics is one of the most fundamental open questions in game theory. In recent work on this front, it was conjectured that the attractors of the replicator dynamic are in one-to-one correspondence with the \emph{sink equilibria} of the game---the sink strongly connected components of a game's preference graph---and it was established that they do stand in at least one-to-many correspondence with them. Here, we show that the one-to-one conjecture is false. We disprove this conjecture over the course of three theorems: the first disproves a stronger form of the conjecture, while the weaker form is disproved separately in the two-player and $N$-player ($N>2$) cases. By showing how the conjecture fails, we lay out the obstacles that lie ahead for characterizing attractors of the replicator, and introduce new ideas with which to tackle them. All three counterexamples derive from an object called a \emph{local source}---a point lying within the sink equilibrium, and yet which is `locally repelling'; we prove that the absence of local sources is necessary, but not sufficient, for the one-to-one property to be true. We complement this with a sufficient condition: we introduce a local property of a sink equilibrium called \emph{pseudoconvexity}, and establish that when the sink equilibria of a two-player game are pseudoconvex then they precisely define the attractors. Pseudoconvexity generalizes the previous cases---such as zero-sum games and potential games---where this conjecture was known to hold, and reformulates these cases in terms of a simple graph property.
\end{abstract}

\section{Introduction} \label{sec:intro}

What are the possible outcomes of a collection of jointly-learning rational agents? This is a fundamental---arguably \emph{the} fundamental---problem in the study of learning in games, with consequences for machine learning, economics and evolutionary biology.  The question has received decades of study by mathematicians, economists and computer scientists \citep[see, for example,][]{zeeman_population_1980,milgrom_adaptive_1991,hofbauer_evolutionary_1996,sandholm2010population,papadimitriou_game_2019} and yet it remains broadly unanswered. Excluding some special cases (such as \emph{zero-sum} and \emph{potential games}, and slight generalizations), we do not know how to compute these outcomes~\citep{sandholm2010population}.

One reason for the lack of progress has been the historical focus on Nash equilibria as the outcome of a game~\citep{papadimitriou_game_2019,myerson1997game}. Over time this approach was found to be problematic \citep{kleinberg_beyond_2011}---not only do learning algorithms fail to converge to Nash equilibria in general games \citep{benaim_perturbations_2012,milionis_impossibility_2023}, but Nash equilibria are also intractable to compute~\citep{daskalakis_complexity_2009}.  

Non-convergence to a Nash equilibria has been traditionally viewed as a limitation of game dynamics. In a departure, \cite{papadimitriou_game_2019} argued the opposite: non-convergence of learning to Nash equilibria should be viewed as yet another limitation of the Nash equilibrium concept. Instead, the outcomes of learning---whatever they may be---should be the fundamental objects of interest in game theory. 
In other words, {\em the meaning of a game $G$ should be understood as a function $\mu_G$ mapping a prior on the space of mixed strategies to a posterior distribution on game outcomes}---where the game outcomes are the attractors of the learning dynamic. To fix a particular well-motivated and widely-used dynamic, we focus on the \emph{replicator dynamic}  \citep{taylor_evolutionary_1978,smith1973logic}---the continuous-time analog of the multiplicative weights algorithm \citep{arora_multiplicative_2012} and the flagship dynamic of evolutionary game theory \citep{hofbauer_evolutionary_2003,sandholm2010population}.  The replicator dynamic is a most natural learning behavior (``move in the direction of the utility gradient''), which is invariant under the addition of positive constants to the players' utilities, and is qualitatively invariant under positive scaling of the utilities---qualities that are \emph{sine qua non} in economic modeling. 

The replicator can have extremely complex behavior \citep{sato_chaos_2002} due to the emergence of chaos and the sensitive dependence on initial conditions, and initially it was not even known whether its attractors are finite in number. However, \cite{papadimitriou_game_2019} suggested a possible path forward: they hypothesized that the behavior of the replicator (and possibly of far more general dynamics) can be approximated by a simple combinatorial tool called here the \emph{preference graph} of the game \citep{biggar_graph_2023,biggar_attractor_2024,biggar_preference_2025}. Specifically, they suggested that the \emph{sink strongly connected components} of this graph---the `\emph{sink equilibria}' \citep{goemans_sink_2005}---are good proxies for the attractors of the replicator.

\begin{figure}
    \centering
    \begin{subfigure}{.45\textwidth}
        \centering
        \begin{tikzpicture}
    \node at (0,0) {\begin{game}{3}{3}[Player~1][Player~2]
         & $a$ & $b$ & $c$ \\
         $a$ & $0,0$ & $2,1$ & $1,2$ \\
         $b$ & $1,2$ & $0,0$ & $2,1$ \\
         $c$ & $2,1$ & $1,2$ & $0,0$ 
         \end{game}
};
\node at (4,0) {};
 \end{tikzpicture}
        \caption{}
        \label{fig:shapley payoff}
    \end{subfigure}
    \qquad
    \begin{subfigure}{.45\textwidth}
        \centering
        \begin{tikzpicture}

\draw [rounded corners, color = gray!30, line width = 10pt] (2,2) to (2,1) to[out = 220, in = -40] (0,1) to (0,0) to (1,0) to[in = -50, out = 50] (1,2) -- cycle;

    \node[nd] (1) at (0,0) {};
    \node[nd] (2) at (1,0) {};
    \node[nd] (3) at (2,0) {};
    \node[nd] (4) at (0,1) {};
    \node[nd] (5) at (1,1) {};
    \node[nd] (6) at (2,1) {};
    \node[nd] (7) at (0,2) {};
    \node[nd] (8) at (1,2) {};
    \node[nd] (9) at (2,2) {};
    
    \draw[arc] (4) to (1);
    \draw[arc] (3) to[in = -40, out = 220] (1);
    \draw[arc] (3) to[out = 50, in = -50] (9);
    \draw[arc] (8) to (9);
    \draw[arc] (5) to (8);
    \draw[arc] (5) to (4);
    
    \draw[arc] (7) to (4);
    \draw[arc] (7) to[in = 130, out = -130] (1);
    \draw[arc] (7) to (8);
    \draw[arc] (7) to[in = 140, out = 40] (9);
    \draw[arc] (1) to (2);
    \draw[arc] (3) to (2);
    \draw[arc] (5) to (2);
    \draw[arc] (2) to[in = -50, out = 50] (8);
    \draw[arc] (3) to (6);
    \draw[arc] (9) to (6);
    \draw[arc] (5) to (6);
    \draw[arc] (6) to[out = 220, in = -40] (4);

 \end{tikzpicture}
        \caption{}
        \label{fig:shapley graph}
    \end{subfigure}
    \caption{The preference graph of \emph{Shapley's game} \citep{shapley_topics_1964}, and a typical payoff matrix representation. It has a unique sink equilibrium, which is the highlighted 6-cycle.}
    \label{fig:shapley}
\end{figure}

This hypothesis is a plausible and hopeful one, first two reasons. First, the preference graph is a natural and insightful graphical representation of the structure of the utilities in a normal-form game \citep{biggar_graph_2023,biggar_preference_2025} (see Figure~\ref{fig:shapley}). Second, the sink equilibria can be viewed as a combinatorial generalization of \emph{pure Nash equilibria} (PNE). Recall that, among the dozens of solution concepts in game theory, only the PNE is largely uncontroversial; its only recognized downside being that not all games have one. Both the sink equilibrium and the mixed Nash equilibrium are generalizations of the PNE which guarantee existence in all games. This observation tempts one to rethink the sink equilibrium as a novel solution concept, a generalization of the PNE in a different direction than Nash's---and possessing the advantages of being both computable in almost linear time in the description of the game, as well as seemingly compatible with learning dynamics.


Unfortunately, proving connections between the sink equilibria and the attractors of the replicator is very difficult.
The first significant progress appeared in \cite{biggar_replicator_2023}, where it was established that every attractor of the replicator must contain one or more sink equilibria, and therefore the sink equilibria are in {\em many-to-one} correspondence with the attractors---implying, for the first time, that the replicator has finitely many attractors. They also articulated the hypothesis of \cite{papadimitriou_game_2019} as a conjecture:
\begin{conj} \label{conj: content conjecture}
    Each attractor of the replicator contains \emph{exactly one} sink equilibrium, and each sink equilibrium is contained in an attractor. 
\end{conj}
This conjecture proposes a strong mathematical correspondence between sink equilibria and attractors.
Biggar and Shames also stated a strictly stronger conjecture, hypothesizing one precise form for this correspondence:
\begin{conj} \label{conj: strong content conjecture}
    The attractors of the replicator in any game are exactly the \emph{content} of the sink equilibria.
\end{conj}
The ``content'' of a set of profiles 
is the union of all subgames spanned by these profiles (see Definition~\ref{def: content}). Were these two conjectures proven true, the quest for the outcomes of learning in games would have come to a triumphant conclusion. Both conjectures {\em are} known to hold in some special cases of sink equilibria \citep{biggar_replicator_2023} such as \emph{attracting subgames} (those where the sink equilibrium profiles are precisely the pure profiles of a subgame). Some other classes of games, such as ordinal potential games \citep{monderer_potential_1996} and weakly acyclic games \citep{young_evolution_1993} have only PNEs as sink equilibria, so all these games also satisfy the conjectures. It was subsequently proved that these conjectures also hold for the sink equilibria of zero-sum games~\citep{biggar_attractor_2024}. 

\subsection*{Our contributions} 
We prove that both Conjecture~\ref{conj: content conjecture} and Conjecture~\ref{conj: strong content conjecture} \emph{fail} in general games. Our results rely on the key concept of a \emph{local source}, which is a mixed profile within a sink equilibrium on the boundary of the strategy space, with the property that nearby trajectories move into the interior of the strategy space, resulting in paths that are not captured by the preference graph. The existence of a local source suffices to disprove Conjecture~\ref{conj: strong content conjecture}; for Conjecture~\ref{conj: content conjecture}, an explicit local source argument suffices to disprove it for \emph{three-player games}.
Disproving Conjecture~\ref{conj: content conjecture} in the case of two-player games requires a more complex argument. In Section~\ref{sec: 2p counter} we analyze a particular $2\times 3$ game (Figure~\ref{fig:2player small}) possessing a local source, and establish additionally the existence of a trajectory between two fixed points. We use this game as a gadget, and by composing several copies we are able to arrive at a two-player game with two sink equilibria, but only one replicator attractor, disproving Conjecture~\ref{conj: content conjecture}.

Finally, we use the insights gained from the three counterexamples to identify a simple local property called \emph{pseudoconvexity}\footnote{Not related to pseudoconvex {\em functions}.}, which suffices to ensure that the content of a sink equilibrium is an attractor (so Conjecture~\ref{conj: strong content conjecture} holds in this case). Pseudoconvexity is a property of the $2\times 2$ subgames which intersect a sink equilibrium, and so it can be easily established computationally. Roughly speaking, pseudoconvexity requires that the net inflow to the sink equilibrium in the neighborhood of a $2\times 2$ subgame that appears ``concave'' is actually positive. This property generalizes all the previous classes of games where Conjecture~\ref{conj: strong content conjecture} was known to hold --- \emph{potential games} \citep{monderer_potential_1996}, games where the sink equilibrium is a subgame \citep{biggar_replicator_2023,ritzberger_evolutionary_1995} and zero-sum games \citep{biggar_attractor_2024} --- as well as some new classes, like \emph{uniform-weighted cycles} (see Section~\ref{sec: understanding pseudoconvexity} and Figure~\ref{fig:shapley}).

We believe that our main contribution is not the fact that the conjectures are false, but rather the tools and concepts through which we are able to disprove them and the new avenues for progress that they reveal. Understanding of local sources and stability conditions like pseudoconvexity are the conceptual obstacles that lie squarely in the path towards a full characterization of the attractors of the replicator dynamic. We note that while the correspondence between attractors and sink equilibria fails in the form articulated in the two conjectures, a broad correspondence between attractors and the combinatorial structure of the preference graph does seem to hold. In Section~\ref{sec: conclusions} we discuss problems inspired and left open by our work, and hypothesize how further developments of the concepts introduced here can build our understanding of learning in games.

\subsection*{Related work} \label{sec: related}

Learning in games has a long and complex history \citep{fudenberg1998theory,cesa2006prediction}. In this paper we focus on the replicator dynamic \citep{taylor_evolutionary_1978}, which is the continuous-time equivalent \citep{sorin2020replicator} of the \emph{multiplicative weights algorithm} \citep{arora_multiplicative_2012,fudenberg1998theory,freund_adaptive_1999}, the flagship method in this field.
The replicator dynamic arose from the work of Maynard Smith on evolutionary game theory \citep{smith1973logic}, being named and formalized in \citep{taylor_evolutionary_1978}. Since then, it has retained its central role in evolutionary game theory \citep{sandholm2010population,hofbauer_evolutionary_2003}. 
Finding its attractors is a central goal of the study of the replicator, both in evolutionary game theory \citep{zeeman_population_1980,sandholm2010population} and more recently in learning \citep{papadimitriou2016nash,papadimitriou_nash_2018}. The preference graph (and related \emph{best-response graph}) have been sporadically rediscovered in game theory \citep{goemans_sink_2005,candogan_flows_2011,pangallo_best_2019,papadimitriou_game_2019,biggar_graph_2023}, see \cite{biggar_preference_2025} for a survey. Sink equilibria originate with the work of \cite{goemans_sink_2005}, who used them to study the Price of Anarchy \citep{koutsoupias1999worst}. Since the work of \cite{papadimitriou_game_2019}, a line of research has developed relating the replicator and the sink equilibria. Recently, \cite{omidshafiei_-rank_2019} used the sink equilibria as an approximation of attractors for the purpose of ranking the strength of game-playing algorithms. Similarly, \cite{omidshafiei_navigating_2020} used the preference graph as a tool for representing the space of games for the purposes of learning. Later, \cite{biggar_graph_2023,biggar_replicator_2023,biggar_attractor_2024} wrote a series of papers on the preference graph and its relationship to the attractors of the replicator dynamic. Another recent work \citep{hakim2024swim} explored the problem of computing the limit distributions over sink equilibria, given a prior over profiles. Our work extends the frontier of this line of investigation.

\subsection*{Preliminaries} \label{sec: prelims}

We study normal-form games with a finite number $N$ of players and finitely many strategies $S_1,S_2,\dots,S_N$ for each player \citep{myerson1997game}. The payoffs in the game are defined by a utility function $u : \prod_{i=1}^{N}S_i \to \real^N$. 
$Z := \prod_{i=1}^{N}S_i$ is the set \emph{strategy profiles} of the game. We often treat the strategy sets $S_i$ implicitly, and define a game simply by the function $u$. We use $p_{-i}$ to denote an assignment of strategies to all players other than $i$. 
A \emph{subgame} is the game resulting from restricting each player to a subset of their strategies. 



A \emph{mixed strategy} is a distribution over a player's pure strategies, and a \emph{mixed profile} is an assignment of a mixed strategy to each player. We often refer to strategies as \emph{pure strategies} and profiles as \emph{pure profiles} to distinguish them from the mixed kinds. For a mixed profile $x$, we write $x^i$ for the distribution over player $i$'s strategies, and $x^i_s$ for the $s$-entry of player $i$'s distribution, where $s\in S_i$. Similar to $p_{-i}$, we use $x_{-i}$ to denote an assignment of mixed strategies to all players other than $i$.
The \emph{strategy space} of the game is the set of mixed profiles, which is given by $\prod_{i=1}^N \Delta_{|S_i|}$ where $\Delta_{|S_i|}$ are the simplices in $\real^{|S_i|}$. We denote this product by $X$, the mixed analog of $Z$. The utility function can be naturally extended to mixed profiles, by taking the expectation over strategies. We denote the utility of a mixed profile $x$ by $\exptu$. Given a subgame $y$, we write $X_y$ (resp. $Z_y$) to denote the set of mixed (resp. pure) profiles in $y$, where (in a slight abuse of notation) we treat $X_y$ as a subset of $X$.


A \emph{Nash equilibrium} of a game is a mixed profile $x$ where no player can increase their expected payoff by a unilateral deviation of strategy. More precisely, a Nash equilibrium is a point where for each player $i$ and strategy $s\in S_i$, $
\exptu_i(x) \geq \exptu_i(s;x_{-i})
$. A Nash equilibrium is \emph{quasi-strict} if the payoff for deviating to strategies outside the support of the equilibrium is \emph{strictly worse}. That is, if $\hat x$ is a Nash equilibrium, it is quasi-strict if $\hat x^i_s = 0$ implies that $\exptu_i(x) > \exptu_i(s;x_{-i})$.

\begin{defn} (Content of a set of pure profiles, \cite{biggar_replicator_2023}.) \label{def: content}
Let $H$ be a set of pure profiles in a game. The content of $H$, denoted $\content(H)$, is the set of all mixed strategy profiles $x$ where all pure profiles in the support of $x$ are in $H$. Equivalently, it is the union $\bigcup_y X_y$ for all subgames $y$ where $Z_y \subseteq H$.
\end{defn}



The \emph{replicator dynamic} is a continuous-time dynamical system~\citep{sandholm2010population,hofbauer_evolutionary_2003} in $X$ defined as the solution of the following ordinary differential equation:
\[
\dot x_s^i = x_s^i\left(\exptu_i(s; x_{-i}) - \sum_{r\in S_i} x_r^i \exptu_i(r; x_{-i}) \right)
\]
, where $x$ is a mixed profile, $i$ a player and $s$ a strategy.  It is known to be the limit of the multiplicative weights update algorithm \citep{arora_multiplicative_2012} when the time step goes to zero. The solution to this equation defines a \emph{flow} \citep{alongi2007recurrence} $\phi : X \times \real \to X$, which is a continuous group action of the reals on $X$. Informally, the flow $\phi(t,x)$ (commonly written $\phi^t(x)$) maps $x\in X$ to the point reached after time $t\in \real$. An \emph{orbit} or \emph{trajectory} of a point $x_0$ is the set $\{ \phi^t(x_0) : t\in \real\}$.

A central notion in dynamical systems is the concept of an attractor \citep{strogatz2018nonlinear}. First, fix a dynamic. An \emph{attracting set} under that dynamic is a set $S$ of points with these two properties: (1) there is a neighborhood $U\supset S, U\neq S$ that is \emph{invariant} under the dynamic (if the dynamic starts in $U$ it will stay there), and (2) all points of $U$ converge to $S$ under the dynamic. An attracting set is an \emph{attractor} if it is minimal, that is, no proper subset of it is attracting. We will also need the notion of the \emph{forward} and \emph{backward} limit sets. Given a point $x$, its \emph{(forward) limit set} $\omega(x)$ under a given dynamic is the set of accumulation points of the orbit starting at $x$. Formally, a point $y$ is an $\omega$-limit point of $x_0$ if there exists a sequence $t_n\in \real$ with $t_n\to\infty$ as $n\to\infty$ with $y = \lim_{t\to\infty} \phi^{t_n}(x_0)$. The forward limit set $\omega(x)$ is the set of all $\omega$-limit points. Similarly, the \emph{backward limit set} $\alpha(x)$ of $x$ is the set of all $\alpha$-limit points, defined analogously as those $y$ where there exists a sequence $t_n\in \real$ with $t_n\to\infty$ as $n\to\infty$ with $y = \lim_{t\to\infty} \phi^{-t_n}(x_0)$ \citep{alongi2007recurrence}.


The \emph{preference graph} of a game \citep{biggar_graph_2023} is a directed graph 
whose nodes are the profiles of the game. Two profiles are $i$-\emph{comparable} if they differ in the strategy of player $i$ only, and they are \emph{comparable} if they are $i$-comparable for some $i$. There is an arc between two profiles if they are comparable, and the arc is directed towards the profile where that player receives higher payoff.  The arcs in the preference graph can be given non-negative weights representing the difference in utility \citep{biggar_graph_2023,biggar_preference_2025}, which gives us the \emph{weighted preference graph}. Specifically, if $p$ and $q$ are $i$-comparable, then the arc $\arc{p}{q}$ is weighted by $u_i(q) - u_i(p) \geq 0$, which we write as $W_{q,p} := u_i(q) - u_i(p)$.
Finally, the  {\em sink  equilibria} of a game are the sink strongly connected components of the preference graph of the game. For a recent summary of results related to preference graphs see \cite{biggar_preference_2025}.

Missing proofs can be found in the appendix.

\section{Refuting the Conjectures} \label{sec: counterexamples}


Our starting point is the result of \cite{biggar_replicator_2023} that every attractor of the replicator dynamic contains some sink equilibrium, 
implying that the attractors of the replicator dynamic are finite in number. Moreover, this result is not too specific to the replicator, as it relies only on two of the dynamic's key properties: \emph{volume conservation} and \emph{subgame-independence} \citep{hofbauer1998evolutionary,sandholm2010population}. Volume conservation prohibits any asymptotically stable set from being in the interior of the strategy space. Subgame-independence asserts that (1) each subgame is invariant under the dynamic and (2) the dynamic in a subgame is unaffected by strategies outside the subgame.
We note that these properties extend to a broad range of variants of the replicator and more complex dynamics, see for example 
\citep{vlatakis-gkaragkounis_no-regret_2020}. 


\subsection{Local sources and Conjecture~\ref{conj: strong content conjecture} }

All three counterexamples are based on a feature a sink equilibrium may or may not contain, called a {\em local source.} 
Consider the game in Figure~\ref{fig:cog game}.
\begin{figure}
    \centering
    \begin{subfigure}{.45\textwidth}
        \centering
        \begin{tikzpicture}

\draw [rounded corners, color = gray!30, line width = 10pt, fill] (1,1) to (1,2) to (2,2) to (2,1) to (1,1) to (0,1) to (0,0) to (1,0) to cycle;

    \node[nd] (1) at (0,0) {};
    \node[nd] (2) at (1,0) {};
    \node[nd] (3) at (2,0) {};
    \node[nd] (4) at (0,1) {};
    \node[nd] (5) at (1,1) {$a$};
    \node[nd] (6) at (2,1) {};
    \node[nd] (7) at (0,2) {$p$};
    \node[nd] (8) at (1,2) {};
    \node[nd] (9) at (2,2) {};

    \node[nd, circle, inner sep= 1pt, fill] (x) at (0.5,1.5) {};
    \node at (0.7,1.7) {$\hat x$};

    \draw[arc,dashed] (5) to (x);
    
    \draw[arc] (4) to (1);
    \draw[arc] (3) to[in = -40, out = 220] (1);
    \draw[arc] (3) to[out = 50, in = -50] (9);
    \draw[arc] (8) to (9);
    \draw[arc] (5) to (8);
    \draw[arc] (5) to (4);
    
    \draw[arc] (7) to (4);
    \draw[arc] (7) to[in = 130, out = -130] (1);
    \draw[arc] (7) to (8);
    \draw[arc] (7) to[in = 140, out = 40] (9);
    \draw[arc] (1) to (2);
    \draw[arc] (3) to (2);
    \draw[arc] (2) to (5);
    \draw[arc] (2) to[in = -50, out = 50] (8);
    \draw[arc] (3) to (6);
    \draw[arc] (9) to (6);
    \draw[arc] (6) to (5);
    \draw[arc] (6) to[out = 220, in = -40] (4);

 \end{tikzpicture} 
        \caption{}
        \label{fig:cog}
    \end{subfigure}
    \qquad
    \begin{subfigure}{.45\textwidth}
        \centering
        \begin{tikzpicture}

\fill [rounded corners, color=gray!30] (-0.2, 0.2) -- (-.2,-0.2) -- (1.7,-0.2) -- (1.7,1.7) -- (1.3,1.7) -- (1.3,0.2) -- cycle;

    \node[nd] (A) at (0,0) {};
    \node[nd] (B) at (0,1.5) {$p$};
    \node[nd] (C) at (1.5,0) {$a$};
    \node[nd] (D) at (1.5,1.5) {};
    \draw[arc] (B) to (A);
    \draw[arc] (B) to (D);
    \draw[arc] (C) to (A);
    \draw[arc] (C) to (D);

    \node[nd, circle, inner sep= 1pt, fill] (x) at (0.75,.75) {};
    \node at (0.8,1) {$\hat x$};

    \draw[arc,dashed] (C) to (x);

\end{tikzpicture}
        \caption{}
        \label{fig:cog local source}
    \end{subfigure}
    \caption{A preference graph (Fig.~\ref{fig:cog}) whose sink equilibrium (highlighted in gray) has a local source $a$ (Def.~\ref{def: local source}). The point $\hat x$ represents the interior fixed point of the $2\times 2$ subgame in the top left, shown separately in Fig.~\ref{fig:cog local source}. The presence of the local source at $a$ implies that any replicator attractor of the game which contains $a$ must also contain $\hat x$, disproving Conjecture~\ref{conj: strong content conjecture}. }
    \label{fig:cog game}
\end{figure}
This game has a unique sink equilibrium $H$ (highlighted in gray), which has an interesting property: despite being \emph{globally} the unique sink connected component of the game, it contains a profile, namely $a$, which is a {\em source} in the upper left subgame in  Fig.~\ref{fig:cog local source}. Such profiles are called \emph{local sources} of the sink equilibrium, a concept that is a basic ingredient of all three counterexamples. 


\begin{defn}[Local sources] \label{def: local source}
    Let $x\in X$ be a mixed profile, $H$ a sink equilibrium, and $Y\subseteq X$ a subgame. Then $x$ is a \emph{local source of $H$ in $Y$} if, (1) $x\in \content(H) \cap Y$, (2) $Y\not\subseteq \content(H)$ and (3) $x$ is a quasi-strict Nash equilibrium of the negated game $-u$ restricted to $Y$.
\end{defn}

In English, this definition can be read as: a local source $x$ of $H$ in $Y$ is a mixed profile in $Y$ which is contained in the sink equilibrium $\content(H)$, yet `locally' (in the subgame $Y$) looks like a \emph{source}, in the sense that all players can strictly improve their payoff by deviating to any currently unplayed strategy in $Y$. 
To illustrate, consider Figure~\ref{fig:cog game}. Let $Y$ be the $2\times 2$ subgame in the upper left, shown separately in Fig.~\ref{fig:cog local source}. The sink equilibrium $H$ intersects $Y$, but $Y\not\subseteq \content(H)$. The profile $a$ is in $\content(H)\cap Y$ it is a source, a PNE of this $2\times 2$ subgame if the utility is negated. Hence $a$ is a local source of $H$ in $Y$. The second requirement ($Y\not\subseteq \content(H)$) is important to eliminate trivial cases; without it, any pure profile in a sink equilibrium that is not a PNE would be a local source, because it is a source in some $1\times 2$ subgame defined by an outgoing arc of the preference graph. 



Armed with the concept of a local source, we return to Conjecture~\ref{conj: strong content conjecture}. In Figure~\ref{fig:cog game}, the node $a$ is a local source of $H$ in the subgame $Y$ (Fig.~\ref{fig:cog local source}). This subgame is a $2\times 2$ coordination game, where we know a trajectory exists\footnote{Note that, when we say a trajectory exists between two fixed points $x$ and $y$, we mean there is a \emph{heteroclinic orbit} between them --- a trajectory where $\alpha(z) = \{x\}$ and $\omega(z) = \{y\}$ for some point $z$.} from each source (in this case $a$) to the Nash equilibrium in the interior of this subgame ($\hat x$) \citep{hofbauer1998evolutionary}. However $\hat x$ is not in the content of the sink equilibrium, because its support includes the profile $p$. The existence of this trajectory implies that any attracting set containing $a$ must also contain $\hat x$, contradicting Conjecture~\ref{conj: strong content conjecture}.
This example is one instance of a more general fact: any sink equilibrium possessing a local source causes Conjecture~\ref{conj: strong content conjecture} to fail. We prove this in the appendix.
\begin{lem} \label{local sources imply instability}
    If $H$ is a sink equilibrium, and $H$ has a local source, then $\content(H)$ is not an attractor. 
\end{lem}


\subsection{Games with (at least) three players} \label{sec: 3p counter}

\begin{figure}
    \centering
    \begin{subfigure}{.45\textwidth}
        \centering
        \begin{tikzpicture}[scale=.7]
\begin{axis}[
    axis line style = {color=white},
    xmin=-1.1, xmax=1.1,
    ymin=-1.1, ymax=1.1,
    zmin=-0.1, zmax=1.1,
    axis equal image,
    xticklabels = {},
    yticklabels = {},
    zticklabels = {},
    xtick style = {draw=none},
    ytick style = {draw=none},
    ztick style = {draw=none},
]



    \draw [rounded corners, color = gray!30, line width = 10pt, fill] (0,0,0) to (1,0,0) to (1,-1,0) to (0,-1,0) to (-1,-1,0) to (-1,-1,1) -- (-1,0,1) -- (0,0,1) -- (0,0,0) -- (0,1,0) -- (-.4,1,0);

    \node[nd,color = gray!30, circle, inner sep= 4pt, fill] at (1,1,1) {};

    
    \node[nd] (hhh) at (0,0,0) {a};
    \node[nd] (thh) at (1,0,0) {};
    \node[nd] (hht) at (0,0,1) {};
    \node[nd] (tht) at (1,0,1) {};
    
    \node[nd] (hth) at (0,1,0) {};
    \node[nd] (tth) at (1,1,0) {};
    \node[nd] (htt) at (0,1,1) {};
    \node[nd] (ttt) at (1,1,1) {b};

    \node[nd] (a) at (-1,0,0) {};
    \node[nd] (b) at (-1,0,1) {};
    \node[nd] (c) at (-1,1,0) {};
    \node[nd] (d) at (-1,1,1) {};
    
    \node[nd] (e) at (0,-1,0) {};
    \node[nd] (f) at (0,-1,1) {};
    \node[nd] (g) at (1,-1,0) {};
    \node[nd] (h) at (1,-1,1) {};

    \node[nd] (i) at (-1,-1,0) {};
    \node[nd] (j) at (-1,-1,1) {};
    
    
    \draw[arc] (hhh) to (hht);
    \draw[arc] (hhh) to (hth);
    \draw[arc] (hhh) to (thh);
    \draw[arc] (htt) to (ttt);
    \draw[arc] (tth) to (ttt);
    \draw[arc] (tht) to (ttt);
    \draw[arc] (htt) to (hht);
    \draw[arc] (tht) to (hht);
    \draw[arc] (htt) to (hth);
    \draw[arc] (tth) to (hth);
    \draw[arc] (tth) to (thh);
    \draw[arc] (tht) to (thh);

    \draw[arc] (h) to (tht);
    \draw[arc] (d) to (htt);
    \draw[arc] (thh) to (g);
    \draw[arc] (g) to (e);
    \draw[arc] (hth) to (c);
    \draw[arc] (c) to (a);
    \draw[arc] (h) to (g);
    \draw[arc] (h) to (f);
    \draw[arc] (d) to (b);
    \draw[arc] (d) to (c);
    \draw[arc] (hht) to (b);
    \draw[arc] (hht) to (f);
    \draw[arc] (f) to (e);
    \draw[arc] (b) to (a);
    \draw[arc] (e) to (hhh);
    \draw[arc] (a) to (hhh);
    \draw[arc] (b) to (j);
    \draw[arc] (j) to (f);
    \draw[arc] (e) to (i);
    \draw[arc] (i) to (a);
    \draw[arc] (i) to (j);
\end{axis}
 \end{tikzpicture}
        \caption{}
        \label{fig:big_3player}
    \end{subfigure}
    \qquad
    \begin{subfigure}{.45\textwidth}
        \centering
        \begin{tikzpicture}[scale=.7]
\begin{axis}[
    axis line style = {color=white},
    xmin=-0.1, xmax=1.1,
    ymin=-0.1, ymax=1.1,
    zmin=-0.1, zmax=1.1,
    axis equal image,
    xticklabels = {},
    yticklabels = {},
    zticklabels = {},
    xtick style = {draw=none},
    ytick style = {draw=none},
    ztick style = {draw=none},
]



    \draw [rounded corners, color = gray!30, line width = 10pt, fill] (0,0,0) to (1,0,0) to (0,0,0) to (0,1,0) to (0,0,0) to (0,0,1) -- cycle;

    \node[nd,color = gray!30, circle, inner sep= 4pt, fill] at (1,1,1) {};


    \node[nd, circle, inner sep= 1pt, fill] (x) at (0.5,0.5,0.5) {};
    \node at (0.6,0.6,0.5) {$\hat x$};
    
    \node[nd] (hhh) at (0,0,0) {$a$};
    \node[nd] (thh) at (1,0,0) {};
    \node[nd] (hht) at (0,0,1) {};
    \node[nd] (tht) at (1,0,1) {};
    
    \node[nd] (hth) at (0,1,0) {};
    \node[nd] (tth) at (1,1,0) {};
    \node[nd] (htt) at (0,1,1) {};
    \node[nd] (ttt) at (1,1,1) {$b$};
    
    \draw[arc] (hhh) to (hht);
    \draw[arc] (hhh) to (hth);
    \draw[arc] (hhh) to (thh);
    \draw[arc] (htt) to (ttt);
    \draw[arc] (tth) to (ttt);
    \draw[arc] (tht) to (ttt);
    \draw[arc] (htt) to (hht);
    \draw[arc] (tht) to (hht);
    \draw[arc] (htt) to (hth);
    \draw[arc] (tth) to (hth);
    \draw[arc] (tth) to (thh);
    \draw[arc] (tht) to (thh);

    \draw[arc,dashed] (hhh) to (x);
    \draw[arc,dashed] (x) to (ttt);
\end{axis}
 \end{tikzpicture}
        \caption{}
        \label{fig:small_3player}
    \end{subfigure}
    \caption{A 3-player counterexample. We show that a replicator trajectory exists from $a$ to $\hat x$ (a fully-mixed Nash equilibrium of the subgame in \ref{fig:small_3player}) and also from $\hat x$ to $b$, implying that any attractor containing $a$ must also contain $b$.}
    \label{fig:3player counterexample}
\end{figure}

Consider a game possessing the preference graph in Figure~\ref{fig:big_3player}. This game has three players\footnote{Note that this counterexample also serves as a counterexample for $N$-player games with $N\geq 3$, because we can embed a copy of this game in a game with more players.} with 3, 3 and 2 strategies respectively. Its preference graph has two sink equilibria $H_a$ and $H_b$, which we have each highlighted in gray. We have named two nodes $a$ and $b$, the former in $H_a$ and the latter in $H_b$. The critical features of the example lie in one $2\times2\times 2$ subgame, which we isolate in Figure~\ref{fig:small_3player}. The remainder of the graph serves to ensure that $H_a$ is a sink equilibrium of the game.

In this subgame, node $a$ is a source and node $b$ is a sink. That is, $a$ is again a local source of $H_a$! Because $b$ is in a different sink equilibrium to $a$, there are necessarily no paths from $a$ to $b$ in this subgraph. However, it is not possible for an attractor to contain $a$ and not $b$.
\begin{lem}
    In Figure~\ref{fig:3player counterexample}, any attracting set of the replicator containing $a$ must contain $b$.
\end{lem}
\begin{proof}
    We will focus on the $2\times 2\times 2$ subgame in Figure~\ref{fig:small_3player} (recall that the replicator is subgame-independent). We define the payoffs in this subgame as one for all players in every sink profile, and zero for all players in every source profile.
    Each player has exactly two pure strategies, and we will represent their mixed strategies by the variables $x_1,x_2,x_3$. Expressed in these variables, we assume w.l.o.g. that $a=(0,0,0)$ and $b=(1,1,1)$. This subgame also contains a Nash equilibrium at $\hat x = (0.5,0.5,0.5)$, which is the only fixed point of the replicator in the interior of the strategy space. Because of the unit payoffs, the replicator equation is given by:
    \[
    \dot x_1 = x_1(1-x_1)\left ((1-x_2)(1-x_3) - x_2(1-x_3) - (1-x_2)x_3 + x_2x_3\right )
    \]
    The equations for $\dot x_2$ and $\dot x_3$ follow by symmetry. Next, consider the one-dimensional diagonal subspace which contains $a$, $b$ and $\hat x$, defined by $x_1=x_2=x_3$. Because of symmetry, this subspace is closed---if we start in this subspace, we remain there. Hence, we can express the replicator on this subspace by a one-dimensional dynamical system, with $w = x_1=x_2=x_3$:
    \[
    \dot w = w(1-w)\left ((1-w)^2 - w(1-w) - (1-w)w + w^2\right )
    \]
    This factorizes to $\dot w = w(1-w)(2w-1)^2$. This equation is always non-negative, with fixed points at 0 ($a$), $0.5$ ($\hat x$) and 1 ($b$). Thus there is a trajectory from any neighborhood of $a$ to $\hat x$, and similarly there is a trajectory from any neighborhood of $\hat x$ to $b$. Any attracting set containing $a$ must also contain $\hat x$, and hence also $b$.
\end{proof}

\subsection{Games with two players} \label{sec: 2p counter}

The previous technique does not work in two-player games. The reason is that every two-player game containing a source and a sink necessarily has a path in the preference graph from the source to the sink. In the example above, our construction used a subgame (Figure~\ref{fig:small_3player}) containing both a source ($a$) and sink ($b$) but with no path between them.

It turns out that the conjecture still fails in two-player games, though the argument is more subtle. We will use the graph in Figure~\ref{fig:2player counterexample}.
\begin{figure}
    \centering
    \begin{subfigure}{.45\textwidth}
        \centering
        \begin{tikzpicture}

\draw [rounded corners, color = gray!30, line width = 10pt] (1,2) to [in = 110, out = -110] (1,-1) -- (0,-1) to [in = -120, out = 120] (0,2) -- (0,1) to[in = 220, out = -40] (2,1) to [in = 50, out = -50] (2,-1) -- (1,-1) -- (0,-1) to [in = -120, out = 120] (0,2) -- cycle;

\node[nd,color = gray!30, circle, inner sep= 4pt, fill] at (3,0) {};

    \node[nd] (1) at (0,0) {};
    \node[nd] (2) at (1,0) {$c$};
    \node[nd] (3) at (2,0) {};
    \node[nd] (4) at (0,1) {};
    \node[nd] (5) at (1,1) {};
    \node[nd] (6) at (2,1) {};
    \node[nd] (7) at (0,2) {$a$};
    \node[nd] (8) at (1,2) {};
    \node[nd] (9) at (2,2) {};

    \node[nd] (10) at (0,-1) {};
    \node[nd] (11) at (1,-1) {};
    \node[nd] (12) at (2,-1) {};

    \node[nd] (14) at (3,0) {$b$};
    \node[nd] (15) at (3,-1) {};
    \node[nd] (16) at (3,1) {};
    \node[nd] (17) at (3,2) {};

    \draw[arc] (3) to (14);
    \draw[arc] (2) to[out = -40, in = -140] (14);
    \draw[arc] (1) to (2);
    \draw[arc] (1) to (4);
    \draw[arc] (1) to (10);
    \draw[arc] (17) to (16);
    \draw[arc] (16) to (14);
    \draw[arc] (15) to (14);
    \draw[arc] (15) to (12);
    \draw[arc] (12) to (11);
    \draw[arc] (11) to (10);
    \draw[arc] (10) to [in = -120, out = 120] (7);
    \draw[arc] (17) to (9);
    \draw[arc] (16) to (6);
    \draw[arc] (6) to [in = 50, out = -50] (12);
    \draw[arc] (3) to (12);
    \draw[arc] (8) to [in = 110, out = -110] (11);
    \draw[arc] (2) to (11);
    
    \draw[arc] (9) to[in = 50, out = -50] (3);
    \draw[arc] (9) to (8);
    \draw[arc] (5) to (8);
    \draw[arc] (5) to (4);
    
    \draw[arc] (7) to (4);
    \draw[arc] (7) to (8);
    \draw[arc] (9) to[out = 140, in = 40] (7);
    \draw[arc] (3) to (2);
    \draw[arc] (5) to (2);
    \draw[arc] (2) to[in = -60, out = 60] (8);
    \draw[arc] (3) to (6);
    \draw[arc] (9) to (6);
    \draw[arc] (5) to (6);
    \draw[arc] (4) to[in = 220, out = -40] (6);

    \node[nd, circle, inner sep= 1pt, fill] (x) at (0.5,1.5) {};
    \node at (0.35,1.3) {$\hat x$};

    \node[nd, circle, inner sep= 1pt, fill] (y) at (1.5,1.5) {};
    \node at (1.7,1.5) {$\hat y$};

    \node[nd, circle, inner sep= 1pt, fill] (z) at (1.5,0.5) {};
    \node at (1.7,0.5) {$\hat z$};

    \draw[arc,dashed] (7) to (x);
    \draw[arc,dashed] (x) to (y);
    \draw[arc,dashed] (y) to (z);
    \draw[arc,dashed] (z) to (2);
    
 \end{tikzpicture}
        \caption{}
        \label{fig:2player big}
    \end{subfigure}
    \quad
    \begin{subfigure}{.45\textwidth}
        \centering
        \begin{tikzpicture}

\draw [rounded corners, color = gray!30, line width = 10pt] (0,2) -- (1,2) -- (0,2) -- (0,1) to[in = 220, out = -40] (2,1) to[out = 220, in = -40] (0,1) -- cycle;

    \node[nd] (4) at (0,1) {};
    \node[nd] (5) at (1,1) {};
    \node[nd] (6) at (2,1) {};
    \node[nd] (7) at (0,2) {$a$};
    \node[nd] (8) at (1,2) {};
    \node[nd] (9) at (2,2) {};
    
    \draw[arc] (9) to (8);
    \draw[arc] (5) to (8);
    \draw[arc] (5) to (4);
    
    \draw[arc] (7) to (4);
    \draw[arc] (7) to (8);
    \draw[arc] (9) to[out = 140, in = 40] (7);
    \draw[arc] (9) to (6);
    \draw[arc] (5) to (6);
    \draw[arc] (4) to[in = 220, out = -40] (6);

    \node[nd, circle, inner sep= 1pt, fill] (x) at (0.5,1.5) {};
    \node at (0.7,1.3) {$\hat x$};

    \node[nd, circle, inner sep= 1pt, fill] (y) at (1.5,1.5) {};
    \node at (1.7,1.5) {$\hat y$};

    \draw[arc,dashed] (7) to (x);
    \draw[arc,dashed] (x) to (y);
    
 \end{tikzpicture}
        \caption{}
        \label{fig:2player small}
    \end{subfigure}
    \caption{A two-player counterexample. We show that replicator trajectories exist from $a$ to $\hat x$ to $\hat y$ to $c$ to $b$, meaning that $b$ must be in any attractor containing $a$.}
    \label{fig:2player counterexample}
\end{figure}
Like before, we have a game containing two nodes $a$ and $b$, where there are no paths from $a$ to $b$ in the graph. The node $a$ is contained in a sink equilibrium $H_a$ and $b$ is contained in a different sink equilibrium $H_b$. We have highlighted these in gray. For clarity we have omitted some arcs from the graph, where they are implied by the fact that $H_a$ and $H_b$ are sink equilibria.

Despite the apparent complexity, the key step can be reduced to reasoning about local sources in a much smaller subgame, shown in Figure~\ref{fig:2player small}. The points $\hat x$, $\hat y$ and $\hat z$ are fixed points of the replicator dynamic in $2\times 2$ subgames --- that is, they are Nash equilibria in their respective subgames. First, consider the $2\times 2$ subgame containing $a$ and $\hat x$. The profile $a$ is a local source of this subgame, and just as we argued in the previous subsection, there is a trajectory from $a$ to $\hat x$, implying that any attracting set contain $a$ must contain $\hat x$. We will now show that, for some choices of arc weights, $\hat x$ is itself a local source of the $2\times 3$ subgame depicted in Fig.~\ref{fig:2player small}. 

\begin{lem} \label{lem: counter 2}
Let $u$ be a (generic) $2\times 3$ game whose preference graph is isomorphic to that in Figure~\ref{fig:2player small}. There exist two fixed points $\hat x$ and $\hat y$ whose support is $2\times 2$, and exactly one of these is a Nash equilibrium. There is a trajectory between these points, beginning at the non-Nash fixed point and ending at the Nash equilibrium.
\end{lem}

Using this lemma, we return to Figure~\ref{fig:2player counterexample}. By choosing appropriate payoffs, we can make $\hat y$ be a Nash equilibrium in the $2\times 3$ subgame containing $\hat x$ and $\hat y$, and make $\hat z$ a Nash equilibrium in the $2\times 3$ subgame containing $\hat y$ and $\hat z$. By Lemma~\ref{lem: counter 2}, there is a trajectory from $\hat x$ to $\hat y$ and similarly there is a trajectory from $\hat y$ to $\hat z$. By the same argument as the case of $a$ and $\hat x$, there is a trajectory from $\hat z$ to $c$. Finally, we complete the argument by observing that these trajectories form a sequence of \emph{heteroclinic orbits} (see footnote 3) from $a$ to $\hat x$ to $\hat y$ to $\hat z$ to $c$ and finally to $b$. This implies that any attracting set which contains $a$ must necessarily contain $b$. Each attractor contains at least one sink equilibrium, and distinct attractors are disjoint. Because any attractor containing $H_a$ must also contain $H_b$, we conclude that this game has only a single attractor, despite having two sink equilibria.

\section{Pseudoconvex sink equilibria are attractors} \label{sec: sink stability}

Our understanding of stability under the replicator dynamic is becoming clearer. When the sink equilibrium has a very simple structure, such as when it is exactly the profiles of a subgame \citep{biggar_replicator_2023}, then its content is an attractor. PNEs are a special case of this. These sinks have no local sources. Sink equilibria of zero-sum games can have non-trivial structure, albeit subject to significant constraints, and they too define attractors \citep{biggar_attractor_2024}. On the other hand, the counterexamples above demonstrate that the presence of local sources can cause some sink equilibria to not be attractors. What special properties of the sink equilibria in these games prevent the occurrence of local sources, and ensure stability? To better understand this, we will examine the structure of their $2\times 2$ subgames, which will lead us to the concept of \emph{pseudoconvexity}.
\subsection{Understanding pseudoconvexity} \label{sec: understanding pseudoconvexity}

In a $2\times 2$ subgame, the presence of a local source requires that exactly three of the profiles are within the sink equilibrium.

\begin{defn} \label{def:corner}
Let $Y$ be a $2\times 2$ subgame of a game. If \emph{exactly three} of the profiles in $Y$ are contained in a sink equilibrium $H$, we call this subgame a \emph{cavity} of $H$. 
\end{defn}
\begin{figure}
    \centering
    \begin{subfigure}{.3\textwidth}
        \centering
        \begin{tikzpicture}
\fill [rounded corners, color=gray!30] (-0.2, 1.7) -- (-.2,1.3) -- (1.3,1.3) -- (1.3,-0.2) -- (1.7,-.2) -- (1.7,1.7) -- cycle;
    \node[nd] (A) at (0,0) {$x$};
    \node[nd] (B) at (0,1.5) {$y$};
    \node[nd] (C) at (1.5,0) {$z$};
    \node[nd] (D) at (1.5,1.5) {$w$};
    \draw[arc] (A) to (B);
    \draw[arc] (B) to node[midway,below] {$b$} (D);
    \draw[arc] (A) to (C);
    \draw[arc] (D) to node[midway,left] {$a$} (C);

\end{tikzpicture}
        \caption{}
        \label{fig:corner 1}
    \end{subfigure}
    \begin{subfigure}{.3\textwidth}
        \centering
        \begin{tikzpicture}

\fill [rounded corners, color=gray!30] (-0.2, 1.7) -- (-.2,1.3) -- (1.3,1.3) -- (1.3,-0.2) -- (1.7,-.2) -- (1.7,1.7) -- cycle;
    \node[nd] (A) at (0,0) {$x$};
    \node[nd] (B) at (0,1.5) {$y$};
    \node[nd] (C) at (1.5,0) {$z$};
    \node[nd] (D) at (1.5,1.5) {$w$};
    \draw[arc] (A) to (B);
    \draw[arc] (B) to node[midway,below] {$b$} (D);
    \draw[arc] (A) to (C);
    \draw[arc] (C) to node[midway,left] {$a$} (D);

\end{tikzpicture}
        \caption{}
        \label{fig:corner 2}
    \end{subfigure}
    \begin{subfigure}{.3\textwidth}
    \centering
    \begin{tikzpicture}

\fill [rounded corners, color=gray!30] (-0.2, 1.7) -- (-.2,1.3) -- (1.3,1.3) -- (1.3,-0.2) -- (1.7,-.2) -- (1.7,1.7) -- cycle;

    \node[nd] (A) at (0,0) {$x$};
    \node[nd] (B) at (0,1.5) {$y$};
    \node[nd] (C) at (1.5,0) {$z$};
    \node[nd] (D) at (1.5,1.5) {$w$};
    \draw[arc] (A) to (B);
    \draw[arc] (D) to node[midway,below] {$b$} (B);
    \draw[arc] (A) to (C);
    \draw[arc] (D) to node[midway,left] {$a$} (C);

\end{tikzpicture}
    \caption{}
    \label{fig:corner 3}
    \end{subfigure}
    \caption{A \emph{cavity} (Def.~\ref{def:corner}) of a sink equilibrium $H$, where $x\not\in H$ and $w,y,z\in H$. Because $x\not\in H$, the arcs at $x$ are necessarily directed towards $y$ and $z$ respectively. Up to symmetry, there are three cases, shown in Figs.~\ref{fig:corner 1},~\ref{fig:corner 2} and~\ref{fig:corner 3}.}
    \label{fig:corners}
\end{figure}

Observe that a sink equilibrium $H$ is a subgame \emph{if and only if} it has no cavities. Equivalently, $H$ is a subgame if and only if $\content(H)$ is convex. Cavities are $2\times 2$ subgames where the sink equilibrium is ``locally concave.'' It is not hard to see that the absence of cavities is sufficient to guarantee that $\content(H)$ is an attractor, but it is a very strong requirement, because it implies the sink equilibrium is a subgame. Can a weaker requirement suffice? Examining Figure~\ref{fig:corners}, we note that case (c) is a local source and so clearly will cause a problem, but case (b) is a local sink --- all points converge uniformly to the sink equilibrium. Finally, in case (a), all points approach the sink equilibrium but may or may not do so uniformly, depending on the relative sizes of $b$ and $a$.

\begin{defn} \label{def: pseudoconvex cavity}
    Let $w = (\alpha,\beta)$ be a profile in a sink equilibrium $H$, and $x = (\gamma,\delta)$ be another profile outside $H$ such that the subgame $Y = \{\alpha,\gamma\}\times\{\beta,\delta\}$ is a cavity of $H$. We say $Y$ is \emph{pseudoconvex} if $W_{(\gamma,\beta),(\alpha,\beta)} + W_{(\alpha,\delta),(\alpha, \beta)} < 0$.
\end{defn}
That is, $Y$ is locally pseudoconvex if the concavity of $Y$ is `not too severe' --- the sum of weights of the arcs into the profile $w$ is positive. Using the language of Figure~\ref{fig:corners}, $Y$ is pseudoconvex if either it is type \ref{fig:corner 2}, or it is type \ref{fig:corner 1} \emph{and} $a\leq b$. This prevents local sources (type \ref{fig:corner 3}) in $2\times 2$ subgames. However, quite remarkably, this is also sufficient to guarantee that the entirety of $\content(H)$ is an attractor (Theorem~\ref{weak source stability} below).
\begin{defn}[Pseudoconvex sink equilibria]
    A sink equilibrium is pseudoconvex if every cavity is pseudoconvex, and a game is pseudoconvex if every sink equilibrium is pseudoconvex.
\end{defn}

Note that because pseudoconvexity depends only on $2\times 2$ subgames, there is a very natural polynomial-time algorithm for checking pseudoconvexity of a sink equilibrium. One must simply examine each cavity and check it satisfies Definition~\ref{def: pseudoconvex cavity}.

As we mentioned earlier, when a sink equilibrium is a subgame it is trivially pseudoconvex, because there are no cavities. It is also true --- though much less obvious, see the appendix --- that the sink equilibria of zero-sum games are pseudoconvex.

\begin{lem} \label{ZS are pseudoconvex}
     Two-player zero-sum games are pseudoconvex.
\end{lem}
Hence for all the classes of sink equilibria where we know the content is an attractor, these sink equilibria are pseudoconvex. It is natural to therefore conjecture that pseudoconvexity is \emph{sufficient} for the content to be an attractor. This turns out to be true for two-player games, as we show in Theorem~\ref{weak source stability}. However, pseudoconvexity also encompasses games which are quite distinct from these cases. As an example, consider the famous example of \emph{Shapley's game}, from \cite{shapley_topics_1964}. A payoff matrix representation is given in Figure~\ref{fig:shapley payoff}, with its preference graph in Figure~\ref{fig:shapley graph}. Shapley demonstrated that \emph{fictitious play} (FP) \citep{brown1949some,robinson_iterative_1951} converged to a 6-cycle on the boundary---which is exactly the content of the unique sink equilibrium. This cycle is obviously far from being a subgame, and further it is also far from being the sink equilibrium of a zero-sum game (see Theorem~4.10 of \cite{biggar_graph_2023}, which proves that such sink equilibria are `\emph{near-subgames}'). However, it is pseudoconvex! First, observe that by Fig.~\ref{fig:shapley payoff} each arc on the 6-cycle has the same weight, which is one. Because the sink equilibrium is exactly a cycle, each cavity (Def.~\ref{def:corner}) is of the form in Figure~\ref{fig:corner 1}. Because the two weights $a$ and $b$ in Fig.~\ref{fig:corner 1}, this sink equilibrium satisfies Definition~\ref{def: pseudoconvex cavity}. More generally, the whole class of {\em uniformly weighted cycles} is pseudoconvex---these are sink equilibrium that are simple cycles where every arc on the cycle has the same weight. Being pseudoconvex, all such cycles turn out to be replicator attractors.

\subsection{Stability of pseudoconvex sink equilibria}

To prove Theorem~\ref{weak source stability}, we will need some new concepts. A key idea is to shift perspective from the mixed strategy space into the \emph{correlated space} of distributions over profiles. We show that the replicator can be given a simple presentation in this space, in terms of a matrix we call the \emph{product matrix} of the game. This construction derives from \cite{biggar_attractor_2024}, who introduced a similar idea in two-player zero-sum games---here we generalize it to $n$-player general-sum games. While straightforward, this transformation is critical for our result and we believe it is a useful fact for analyzing the replicator in other contexts as well.

\begin{lem} \label{product replicator}
Let $u$ be a $N$-player $m_1\times m_2 \times \dots\times m_N$ game. The \emph{product matrix} of $u$ is the following matrix $M\in \real^{(\prod_i m_i)\times (\prod_i m_i)}$, indexed by profiles $p$ and $q$ in $Z$:
\begin{equation} \label{product matrix}
M_{q,p} = \sum_{i=1}^N \left ( u_i(q_i;p_{-i}) - u_i(p) \right ) = \sum_{i=1}^N W_{(q_i;p_{-i}),p}
\end{equation}
Next, let be $p = (s_1,s_2,\dots,s_N)$ be a pure profile. Given a mixed profile $x$, the product distribution $z_p := \prod_{i=1}^N x^i_{s_i}$ defines a distribution over profiles. Then, under the replicator dynamic:
\[ \dot z_p = z_p (M z)_p \]
\end{lem}

In other words, the distribution $z$ induced over profile in $Z$ by a mixed strategy $x$ evolves by a simple formula in terms of the product matrix---simpler even than the original definition of the replicator! Using this, we can prove the main theorem of this section:


\begin{thm}\label{weak source stability}
    If a sink equilibrium $H$ of a two-player game $u$ is pseudoconvex, then its content is an attractor of the replicator.
\end{thm}
\begin{proof}[\emph{Sketch---full proof in appendix}]
We prove this using a \emph{Lyapunov argument} on the cumulative total mass distributed over the profiles in $H$. In correlated space, this is simply the sum $z_H := \sum_{h\in H} z_h$ where $z$ is the mass on a single profile, as in Lemma~\ref{product replicator}. The total $z_H$ is equal to one exactly when a point lies in $\content(H)$. We show that this is increasing near $\content(H)$ (where $z_H$ is sufficiently close to one). Using linearity and Lemma~\ref{product replicator}, we can obtain an expression for $\dot z_H$ in terms of the product matrix $M$, where each term corresponds to an arc in the preference graph with at least one endpoint in $H$. By grouping terms into $2\times 2$ subgames, we end up with a case-by-case analysis on the preference graphs of those subgames. With $z_H$ close to one, we show that the only potentially negative term is caused by a cavity that is not pseudoconvex! Applying pseudoconvexity thereby guarantees the result.
\end{proof}

We believe that Theorem~\ref{weak source stability} is a major step towards our ultimate goal: a polynomial-time algorithm which, given a two-player game in normal form, identifies its attractors. When the sink equilibria are pseudoconvex, we know the attractors: they are the content of these sink equilibria. But if a local source exists in a sink component, some replicator paths ``escape" the component. However finding local sources, or determining if any exist, can be very hard. The power of this theorem is that pseudoconvexity (a $2\times 2$ property) is sufficient to guarantee, non-constructively, that no such escaping trajectories exist. If the sink equilibrium is not pseudoconvex, the proof of the theorem does not seem to provide guidance on where to look for these escaping trajectories. Analyzing the attractors of two-player games beyond pseudoconvexity is an important open research problem left by this work.

\section{Conclusions and open problems} \label{sec: conclusions}

Let us return to our original goal: to understand, and ultimately compute, the possible long-run outcomes of game dynamics by exploiting their relationship to sink equilibria. We have made some significant progress on this problem for the replicator dynamic. For games where the sink equilibria are pseudoconvex, they give a precise and efficiently-computable characterization of the attractors. On the other hand, we have also shown how this simple picture does not apply to all games, as the presence of local sources can lead to attractors that are larger than the sink equilibria, sometimes merging two or more sink equilibria. 


This work opens up a number of important open problems.

\paragraph{Local sources and paths.} We know that the presence of local sources can lead to trajectories which escape sink equilibria. Yet local sources too possess a specific graph-theoretic structure---we believe that a broader combinatorial framework generalizing the preference graph could incorporate this case, and thus potentially characterize the ultimate structure of the attractors of a game.



\paragraph{Beyond uniform convergence.} The criteria for a set to be an attractor are very strong---in particular they require a neighborhood of the set to exist where all trajectories approach the attractor. However, for the replicator dynamic, there exist cases where this requirement seems too demanding. Figure~\ref{fig:2x3 dominance} shows a $2\times 3$ game, where the second column strategy dominates the third. Under the replicator, the third column eventually vanishes from all interior starting points, and so the dynamic converges to the highlighted $2\times 2$ subgame. But the preference graph is strongly connected, so the unique attractor is the whole game. The discrepancy is explained by the fact that the replicator does not converge \emph{uniformly} to this smaller set \citep{biggar_attractor_2024,biggar_preference_2025}. This lack of uniformity follows as a consequence of subgame-independence. This suggests that stronger prediction of game dynamics could be made if we allowed for weaker notions of convergence.
\begin{figure}
    \centering
    \begin{tikzpicture}

    \node[nd] (1) at (0,0) {};
    \node[nd] (2) at (0,1) {};
    \node[nd] (3) at (1,0) {};
    \node[nd] (4) at (1,1) {};
    \node[nd] (5) at (2,0) {};
    \node[nd] (6) at (2,1) {};

    \draw [rounded corners, color = gray!30, line width = 10pt] (0,0) to (0,1) to (1,1) to (1,0) to cycle;
    
    \draw[arc] (1) to (3);
    \draw[arc] (1) to[out = -50, in = -130] (5);
    \draw[arc] (6) to[out = 130, in = 50] (2);
    \draw[arc] (5) to (3);
    \draw[arc] (4) to (2);
    \draw[arc] (6) to (4);
    
    \draw[arc] (5) to (6);
    \draw[arc] (3) to (4);
    \draw[arc] (2) to (1);
 \end{tikzpicture}
    \caption{Under the replicator, all interior starting points converge to the highlighted $2\times 2$ subgame. However, this is a strict subset of the unique attractor, which is the whole strategy space \citep{biggar_preference_2025}.}
    \label{fig:2x3 dominance}
\end{figure}

\paragraph{Algorithmic problems.} The ultimate goal of this research program is a polynomial-time algorithm which, given a game in normal form, will output the combinatorial structure of its attractors, and do this for a wide range of learning dynamics. For the case of the replicator dynamic we have made reasonable progress, with pseudoconvexity providing a simple sufficient condition for local stability. However, the possibility of local sources leads to two new problems:
\begin{enumerate}
    \item We know that the presence of a local source in a sink equilibrium $H$ is sufficient to show that $\content(H)$ is not an attractor. But what of the converse: if a sink equilibrium is not stable, must there exist a local source acting as a `witness' to this instability? From a computational perspective, given $H$, can we efficiently verify if $\content(H)$ is an attractor?

    \item If a local source does exist, some additional points may need to be added to the sink equilibrium. But which? Is there an iterative procedure for adding the missing points to a sink equilibrium until we find an attractor?
\end{enumerate}

\paragraph{Large games.}  The limit behavior of learning in {\em large} multi-player games is of great interest in economics --- however, we know that even the most modest algorithmic goals related to sink equilibria are PSPACE-complete for many succinct representations of games \citep{fabrikant2008complexity}. It would be very interesting to make progress in characterizing and computing the attractor structure of a game for the case of {\em symmetric} multi-player games.

\section*{Acknowledgements}
This work supported by a grant by the NSF, and a MURI grant from the Office of Naval Research.

\bibliographystyle{plainnat}
\bibliography{refs}

\appendix

\section{Omitted proofs}

\begin{lem}[Lemma~\ref{local sources imply instability}]
    If $H$ is a sink equilibrium, and $H$ has a local source, then the $\content(H)$ is not an attractor. 
\end{lem}
\begin{proof}
    Let $x$ be a local source of $H$ in a subgame $Y$. Because $Y\not\subseteq \content(H)$, we know that $\content(H)\cap \intr Y = \emptyset$. Thus we only need to show that at there is at least one point $z\in\intr Y$ with $\alpha(z) = \{x\}$. Because $x$ is a quasi-strict Nash equilibrium of $-u$, this is equivalent to showing that, in $-u$, there is a $z\in\intr Y$ with $\omega(z) = \{x\}$. We prove this claim in Lemma~\ref{nash interior stable points}, by observing that the stable manifold of $x$ in $-u$ must necessarily intersect the interior of $X$.
\end{proof}

\begin{lem} \label{nash interior stable points}
    Let $\hat x$ be a quasi-strict Nash equilibrium on the boundary $\bd X$ of the strategy space $X$. Then there exists a point $z\in\intr X$ with $\omega(z) = \{\hat x\}$.
\end{lem}
\begin{proof}
    Let $Y$ be the subgame that is the support of $\hat x$, with does not include all strategies because $\hat x$ lies in the boundary of $X$. We begin by constructing the Jacobian of the replicator at $\hat x$. Let $s_i$ be a strategy that is not in the support of $\hat x^i$, so $\hat x^i_{s_i} = 0$. Observe that
    \[
    \partial \dot x^i_{s_i} / \partial x^i_{s_i} = u_i(s_i;x_{-i}) - u_i( x) +  x^i_{s_i}\partial /\partial x^i_{s_i} (u_i(s_i; x_{-i}) - u_i(x))
    \]
    and so at $\hat x$, where $\hat x^i_{s_i} = 0$, we have $\partial \dot x^i_{s_i} / \partial x^i_{s_i} = u_i(s_i;\hat x_{-i}) - u_i(\hat x)$. Similarly, for any other $x^k_\ell$,
    \[
    \partial \dot x^i_{s_i} / \partial x^k_\ell = x^i_{s_i}( \dots)
    \]
    so $\partial \dot x^i_{s_i} / \partial x^k_\ell = 0$ at $\hat x$. Hence the strategy $s_i$ column of the Jacobian of the replicator at $\hat x$ is equal to $(u_i(s_i;\hat x_{-i}) - u_i(\hat x))\mathbf{e}_{s_i}$, where $\mathbf{e}_{s_i}$ is the standard basis vector of the $s_i$ column. We conclude that $\mathbf{e}_{s_i}$ is an eigenvector of the Jacobian at $\hat x$, whose eigenvalue is $u_i(s_i;\hat x_{-i}) - u_i(\hat x)$ --- the relative payoff of $s_i$ at $\hat x$. Because this is a quasi-strict Nash equilibrium, this is strictly negative. These are called the transversal eigenvalues of $\hat x$ \citep{hofbauer1998evolutionary}.

    We now apply the stable manifold theorem. This tells us that, in the neighborhood of $\hat x$, there is an invariant manifold of points which converge to $\hat x$, and this manifold is tangent to the space $E^s$ spanned by the eigenvectors with negative eigenvalues. Because all transversal eigenvalues are negative, there is a direction in the stable manifold of $\hat x$ which points into the interior of the space. Hence there exists a points on the interior of $X$ which converges to $\hat x$.
\end{proof}

\begin{lem}[Lemma~\ref{lem: counter 2}]
Let $u$ be a generic $2\times 3$ game whose preference graph is isomorphic to that in Figure~\ref{fig:2player small}. There exist two fixed points $\hat x$ and $\hat y$ whose support is $2\times 2$, and exactly one of these is a Nash equilibrium. There is a trajectory between these points, beginning at the non-Nash fixed point and ending at the Nash equilibrium.
\end{lem}
\begin{proof}
(\emph{Claim:} generically, exactly one of $\hat x$ and $\hat y$ is Nash.) Let $r$ be the pure strategy for player 1 which involves playing the top row. Then $\hat x^1_r,\hat y^1_r \in (0,1)$ denote the probability mass on $r$ in $\hat x$ and $\hat y$ respectively. We will show that $\hat x$ is Nash if $\hat x^1_r > \hat y^1_r$, $\hat y$ is Nash if $\hat x^1_r < \hat y^1_r$ and both are Nash in the non-generic case where $\hat x^1_r = \hat y^1_r$.

Let $c_1,c_2,c_3$ be the column strategies, ordered from left to right. Given a fixed mixed strategy $(z,1-z)$ for the row player, the column player receives a payoff for each $c_i$ that is a linear function of $z$. At $z = 0$, the column player prefers their strategies in the best-to-worst order $c_3,c_1,c_2$. At $z = 1$, the column player prefers their strategies in the best-to-worst order $c_2,c_1,c_3$. At $z = \hat x^1_r$, the column player is indifferent between $c_1$ and $c_2$ and at $z = \hat y^1_r$, they are indifferent between $c_2$ and $c_3$. If $\hat x^1_r < \hat y^1_r$, then at $z = \hat x^1_r$ the column player must prefer $c_3$ over $c_1$ and $c_2$, so $\hat x$ is not a Nash equilibrium. By contrast, at $z = \hat y^1_r$, the column player prefers both $c_2$ and $c_3$ (and hence $\hat y^2$) over $c_1$, so $\hat y$ is a Nash equilibrium. The opposite is true when $\hat x^1_r > \hat y^1_r$. If $\hat x^1_r = \hat y^1_r$, then at this point the column player is indifferent between $c_1$, $c_2$ and $c_3$, resulting in a continuum of Nash equilibria. This completes the claim.

\begin{figure}
    \centering
    \begin{tikzpicture}

    \draw [rounded corners, color = gray!30, line width = 5pt] (1.5,1.5) -- (2.25,2.25) -- (3,3) to[out = 135, in = 45] (0,3) -- cycle;

    \draw [rounded corners, color = gray!30, line width = 5pt, fill] (1.5,1.5) -- (0,1.5) -- (0,3) -- (0.75,2.25) -- cycle;

    \node[nd] (4) at (0,1.5) {$d$};
    \node[nd] (5) at (1.5,1.5) {$r_1$};
    \node[nd] (6) at (3,1.5) {$s_2$};
    \node[nd] (7) at (0,3) {$a$};
    \node[nd] (8) at (1.5,3) {$s_1$};
    \node[nd] (9) at (3,3) {$r_2$};
    
    \draw[arc] (9) to (8);
    \draw[arc] (5) to (8);
    \draw[arc] (5) to (4);
    
    \draw[arc] (7) to (4);
    \draw[arc] (7) to (8);
    \draw[arc] (9) to[out = 140, in = 40] (7);
    \draw[arc] (9) to (6);
    \draw[arc] (5) to (6);
    \draw[arc] (4) to[in = 220, out = -40] (6);

    \node[nd, circle, inner sep= 1pt, fill] (x) at (.75,2.25) {};
    \node at (0.75,2) {$\hat x$};

    \node[nd, circle, inner sep= 1pt, fill] (y) at (2.25,2.25) {};
    \node at (2.25,2) {$\hat y$};
    
 \end{tikzpicture}
    \caption{The game from Figure~\ref{fig:2player small}, with the set $W$ highlighted as in the proof of Lemma~\ref{lem: counter 2}.}
    \label{fig:proof explanation}
\end{figure}
(\emph{Claim:} there exists a point $z$ where $\alpha(z) = \hat x$ and $\omega(z) = \hat y$.) In Figure~\ref{fig:proof explanation} we show a more detailed depiction of this game for the purposes of the proof. The game has two sources and two sinks, which we name $r_1$, $r_2$, $s_1$ and $s_2$ respectively. The remaining two profiles are named $a$ and $d$. The sinks are PNEs, both of which are attractors of the replicator. Further, because every attractor must contain a sink equilibrium \citep{biggar_replicator_2023}, these are the only attractors.

Let $\omega^{-1}(s_1)$ and $\omega^{-1}(s_2)$ be the sets of points $z$ whose limit set $\omega(z)$ is $\{s_1\}$ or $\{s_2\}$ respectively. These are disjoint invariant sets, and they are open by the continuity of the replicator dynamic. Hence there exists a closed set $W$ of points whose limit sets are neither $\{s_1\}$ nor $\{s_2\}$. Note that $W$ is also an invariant set. 
Because each proper subgame of this game is two-dimensional, we can completely characterize the intersection of $W$ and the boundary $\bd X$ of the strategy space by considering each subgame in turn. The result is the set depicted in Figure~\ref{fig:proof explanation}. It contains all fixed points except the sinks $s_1$ and $s_2$, as well as the one-dimensional trajectories joining them, and a two-dimensional region in the subgame spanned by $a$, $\hat x$, $r_1$ and $d$, and the trajectories between them. 

By Proposition~5.1 of \cite{hofbauer_time_2009}, the limit set of the time average of an interior point $z$ is a closed subset of $X$ that is both invariant under \emph{best-response dynamics}. For $z\in \intr X \cap W$, $\omega(z)$ is on the boundary \citep{hofbauer1998evolutionary}, and so must additionally be a subset of $W \cap \bd X$. The only subset of this set that is invariant under best-response dynamics is the Nash equilibrium set $\{\hat y\}$. Hence for any $z\in W\cap \intr X$, $\omega(z) = \{\hat y\}$.

Now we repeat the steps above, but with the role of forward and backward limit sets reversed. Let $S$ be the set of points $z$ in $X$ whose \emph{backward} limit set $\alpha(z)$ is neither $\{r_1\}$ nor $\{r_2\}$. Again we can characterize the set $S\cap \bd X$, which has an analogous structure to $W$. The flow of the replicator in reverse time is equal to the flow of the replicator on the negated game $-u$, so we can apply the same arguments about $\omega$-limits to $\alpha$-limits. We conclude, again, that for every $z\in S\cap \intr X$, $\alpha(z) = \{\hat x\}$.

To complete the argument, we just need to show that $S\cap W\cap \intr X$ is not empty. For this, observe that if $z\in \bd X \cap \omega^{-1}(s_1)$, then because $\omega^{-1}(s_1)$ is open there is an $\epsilon > 0$ such that all points within $\epsilon$ of $z$ are in $\omega^{-1}(s_1)$. In particular, there are points near $z$ in $\intr X \cap \omega^{-1}(s_1)$.

Now consider a neighborhood of $\hat x$ in $S\cap \intr X$. $\hat x$ lies on the boundary of $\omega^{-1}(s_1)$ and $\omega^{-1}(s_2)$, so every neighborhood in $S$ contains points in both of these sets. Similarly, every neighborhood in $S\cap \intr X$ must contain points in both of these sets, by the argument in the previous paragraph. Finally, every neighborhood of $\hat x$ in $S\cap \intr X$ must therefore contain points in the boundary of $\omega^{-1}(s_1)$ and $\omega^{-1}(s_1)$, which is entirely contained in $W$. Hence $S\cap W\cap \intr X$ is nonempty.
\end{proof}

\begin{lem}[Lemma~\ref{ZS are pseudoconvex}]
     Two-player zero-sum games are pseudoconvex.
\end{lem}
\begin{proof}
     Let $w = (\alpha,\beta)$ be a profile in a sink equilibrium $H$, and $x = (\gamma,\delta)$ be another profile outside $H$ such that the subgame $\{\alpha,\gamma\}\times\{\beta,\delta\}$ is a cavity of $H$. Because $x\not\in H$, there are arcs $\arc{(\gamma,\delta)}{(\gamma,\beta)}$ and $\arc{(\gamma,\delta)}{(\alpha, \delta)}$ and so $W_{(\gamma,\beta),(\gamma,\delta)} > 0$ and $W_{(\alpha, \delta),(\gamma,\delta)} > 0$. Lemma~4.7 of \cite{biggar_attractor_2024} establishes that $W_{(\gamma,\beta),(\gamma,\delta)}  + W_{(\alpha, \delta),(\gamma,\delta)} =  -(W_{(\gamma,\beta),(\alpha, \beta)} + W_{(\alpha,\delta),(\alpha, \beta)})$, so $W_{(\gamma,\beta),(\alpha,\beta)} + W_{(\alpha,\delta),(\alpha, \beta)} < 0$, thus this cavity is pseudoconvex.
\end{proof}

\begin{lem}[Lemma~\ref{product replicator}]
Let $u$ be a $N$-player $m_1\times m_2 \times \dots\times m_N$ game. The \emph{product matrix} of $u$ is the following matrix $M\in \real^{(\prod_i m_i)\times (\prod_i m_i)}$, indexed by profiles $p$ and $q$ in $Z$:
\begin{equation*}
M_{q,p} = \sum_{i=1}^N \left ( u_i(q_i;p_{-i}) - u_i(p) \right ) = \sum_{i=1}^N W_{(q_i;p_{-i}),p}
\end{equation*}
Next, let be $p = (s_1,s_2,\dots,s_N)$ be a pure profile. Given a mixed profile $x$, the product distribution $z_p := \prod_{i=1}^N x^i_{s_i}$ defines a distribution over profiles. Then, under the replicator dynamic:
\[ \dot z_p = z_p (M z)_p \]
\end{lem}
\begin{proof}
By Lemma~A.1 of \cite{biggar_replicator_2023}, the replicator equation can be written
\begin{align*}
    \dot x^i_s &= x_s^i\sum_{r\in S_i} x_r^i \sum_{p_{-i} \in Z_{-i}} z_{p_{-i}} \left (u_i(s;p_{-i}) - u_i(r;p_{-i}) \right ) \\
    &= x_s^i\sum_{p_{-i} \in Z_{-i}} \sum_{r\in S_i} x_r^iz_{p_{-i}}\left (u_i(s;p_{-i}) - u_i(r;p_{-i}) \right ) \\
    &= x_s^i\sum_{p_{-i} \in Z_{-i}} \sum_{r\in S_i} z_{(r;p_{-i})} \left (u_i(s;p_{-i}) - u_i(r;p_{-i}) \right ) \\
    &= x_s^i\sum_{p\in Z} z_p \left (u_i(s;p_{-i}) - u_i(p) \right)
\end{align*}

Now we observe that for $q = (q_1,q_2,\dots,q_N)$,
\begin{align*}
    \dot z_q &= \ddt (\prod_{i=1}^N x^i_{q_i})
    = z_q \sum_{i=1}^N \frac{\dot x^i_{q_i}}{x^i_{q_i}} \quad \text{(product rule)} \\
    &= z_q \sum_{i=1}^N  \sum_{p\in Z} z_p \left (u_i(q_i;p_{-i}) - u_i(p) \right) \quad \text{(by above)}\\
    &= z_q \sum_{p\in Z} z_p \sum_{i=1}^N \left ( u_i(q_i;\bar p) - u_i(p) \right) 
    =  z_q \sum_{p\in Z} z_p M_{q,p} \\
    &=  z_q (M z)_q
\end{align*}
\end{proof}

\begin{thm}[Theorem~\ref{weak source stability}]
    If a sink equilibrium of a two-player game is pseudoconvex, then its content is an attractor of the replicator.
\end{thm}
\begin{proof}
    Let $H$ be a sink equilibrium, which we assume is pseudoconvex. We will show it is asymptotically stable by a Lyapunov argument, using a similar structure to Theorem~4.3 of \cite{biggar_attractor_2024}. First, we define $z_H := \sum_{h\in H} z_h$. That is, $z_H$ is the cumulative total distributed over the profiles in $H$ in the product distribution $z$. Note that $z_H = 1$ if and only if $x \in \content(H)$. $z_H$ is uniformly continuous, and so to prove that $\content(H)$ is an attractor it is sufficient to show that $\dot z_H > 0$ in some neighborhood of $\content(H)$. Now fix some small $1> \epsilon > 0$. We will assume that $z_H = 1 - \epsilon$, and we will show that for small enough $\epsilon$, $\dot z_H > 0$.

    From Lemma~\ref{product replicator} we have that $\dot z_H = \sum_{h\in H} z_h(M z)_h = \sum_{p\in Z}\sum_{h\in H} z_p z_h M_{h,p}$. Each term in this sum corresponds to a pair of profiles $p$ and $h$ with $h\in H$. First, we divide this sum into comparable and non-comparable pairs of profiles:
    \[
    \dot z_H  = \sum_{p,h\in H \text{ comparable}} z_p z_h M_{h,p} + \sum_{p,h\in H \text{ not comparable}} z_p z_h M_{h,p}
    \]
    If $p$ and $h$ are comparable and $p\in H$, then the first sum contains the terms $z_pz_h M_{p,h}$ and $z_pz_h M_{h,p}$. Because they are comparable, $M_{p,h} = u_i(p) - u_i(h) = W_{p,h} = - W_{h,p} = - M_{h,p}$ (by equation~\eqref{product matrix}), so these terms cancel. The sum becomes
    \[
    \dot x_H  = \sum_{p\not\in H,h\in H \text{ comparable}} z_p z_h M_{h,p} + \sum_{p,h\in H \text{ not comparable}} z_p z_h M_{h,p}
    \]
    If $p$ and $h$ are comparable, and $p\not\in H$ and $h\in H$, $M_{h,p} > 0$ because the arc $\arc{p}{h}$ is necessarily directed into $h$ ($H$ is a sink component). Hence all terms in the first sum are non-negative.
    
    Now consider the case where $p$ and $h$ are non-comparable. Suppose $p = (\alpha,\beta)$ and $h= (\gamma,\delta)$ where none of these are equal, because $p$ and $h$ are not comparable. Let $a := (\alpha,\delta)$ and $b := (\gamma,\beta)$ be the other two profiles in this $2\times 2$ subgame. First observe that $z_p z_h = x_\alpha y_\beta x_\gamma y_\delta = z_a z_b$. So, for instance, if $z_pz_hM_{h,p}$ and $z_az_bM_{a,b}$ are both in the sum, we can combine them into a single term $z_pz_h(M_{h,p} + M_{a,b})$. We will now group all such terms with common $z$-coefficients. Which of these pairs actually appear in the sum depends on which of $p,a,b,h$ are in $H$. We know $h\in H$, so the remaining cases are:
    \begin{enumerate}
        \item All of $p,a,b,h$ are in $H$: The sum contains all of the terms $M_{p,h} + M_{h,p} + M_{a,b} + M_{b,a}$. Expanding this by Definition~\ref{product matrix} gives $W_{a,p} + W_{b,p} + W_{a,h} + W_{b,h} + W_{p,b} + W_{h,b} + W_{h,a} + W_{p,a} = 0$, because each term $W_{i,j} = -W_{j,i}$. Hence if all are in $H$, these terms cancel in the sum.
        \item Three are in $H$ --- we assume w.l.o.g. that $p\not\in H$ and $a,b,h\in H$. This is a cavity of $H$. Then the terms in the sum are $M_{h,p} + M_{a,b} + M_{b,a}$. By the same argument as above, $M_{h,p} + M_{a,b} + M_{b,a} = -M_{p,h}$, which is equal to $M_{p,h} = W_{a,h} + W_{b,h}$. Hence, $M_{p,h}$ is positive \emph{if and only if} this cavity is pseudoconvex. By assumption, all cavities are pseudoconvex, so $-W_{p,h} = M_{h,p} + M_{a,b} + M_{b,a}$ is non-negative.
        \item Two are in $H$---assume w.l.o.g. that $p,a\not\in H$ and $b,h\in H$. The sum therefore contains the terms $M_{h,p}$ and $M_{b,a}$. By Definition~\ref{product matrix}, $M_{h,p} + M_{b,a} = W_{p,a} + W_{h,a} + W_{a,p} + W_{b,p} = W_{h,a} + W_{b,p}$. Since $b$ and $h$ are in $H$ and $q$ and $b$ are not, the arcs $\arc{a}{h}$ and $\arc{p}{b}$ must be directed into the component, so $W_{h,a} > 0$ and $W_{b,p} > 0$. Hence this term is also non-negative.
        \item Only $h$ is in $H$. This last case is the most difficult, because here the terms can be negative. Each such pair contains only the term $z_pz_h M_{h,p}$. However, note that because $a,b \not \in  H$, $z_a,z_b <\epsilon$. Hence $z_pz_h = z_az_b < \epsilon^2$. Thus all negative terms in $\dot x_H$ have coefficient at most $\epsilon^2$.
    \end{enumerate}

    We have now grouped the terms in this sum by their distinct $z$-coefficients, so each term has the form $x_\alpha y_\beta x_\gamma y_\delta(M_{i,j} + \dots)$. For simplicity, we write $K_{\alpha,\beta,\gamma,\delta} := \sum M_{i,j}$ where $z_iz_j = x_\alpha y_\beta x_\gamma y_\delta$. We now define $\mu := \min \{ K_{\alpha,\beta,\gamma,\delta}\ :\ K_{\alpha,\beta,\gamma,\delta} > 0\}$ and similarly $m := \max \{ |K_{\alpha,\beta,\gamma,\delta}|\ :\ K_{\alpha,\beta,\gamma,\delta} < 0\}$. These are the smallest and largest positive and negative terms respectively. There are at most $N^2$ terms in this sum, where $N$ is the number of profiles. By the above, each negative term has coefficient at most $\epsilon^2$, so the total sum of negative terms is at most $-mN^2\epsilon^2$. 
     Now select an $h\in H$ where $z_h \geq (1 - \epsilon)/|H|$. Since $z_H = 1-\epsilon$, such a node must exist. Then:
     \[
    \dot x_H \geq \sum_{p\not\in H,\ z_pz_h = x_\alpha y_\beta x_\gamma y_\delta,\ K_{\alpha,\beta,\gamma,\delta}>0} z_pz_h \mu - mN^2\epsilon^2 = \mu z_h \sum_{p\not\in H,\ z_pz_h = x_\alpha y_\beta x_\gamma y_\delta,\ K_{\alpha,\beta,\gamma,\delta}>0} z_p - mN^2\epsilon^2
    \]
    This inequality holds because we have retained the contribution from all negative terms (in the $mN^2\epsilon^2$ term), and reduced the set of positive terms. Specifically we have included terms where the coefficient equals $z_pz_h$ for our fixed $h$ and where $p\not\in H$. Now we must determine the sum $\sum z_p$ over these $p$. The total sum over $z_p$ with $p\not\in H$ is $\epsilon$, but some $z_p$ are not included in this sum, if they correspond to a negative term $z_wz_h K_{\alpha,\beta,\gamma,\delta} < 0$. However, by the argument above, this occurs only in case (4). There, the remaining profiles $a$ and $b$ are not in $H$, and so contribute two positive terms $z_az_h W_{h,a}$ and $z_bz_h W_{h,b}$ to this sum. By this argument, at least 2/3 of the terms in this sum must be positive. Also, each such $z_p$ has $z_p <\epsilon^2$. We obtain 
    \begin{align*}
        \dot z_H
        &\geq  k ((1-\epsilon)/|H|)(\epsilon - (1 - |H|)\epsilon^2/3) - mN^2\epsilon^2 \\
        &\geq k\epsilon/|H| - o(\epsilon^2)
    \end{align*}
    Thus, for small enough $\epsilon>0$, this term is strictly positive. 
\end{proof}

\end{document}